\documentclass[12pt]{article}
\usepackage{graphicx,psfrag,epsf}
\usepackage{enumerate,color,float}

\usepackage{url} 
\usepackage{amssymb, amsthm, amsmath}
\usepackage{graphicx}
\usepackage[authoryear]{natbib}
\usepackage{bm}
\usepackage{comment}
\newcommand{\bdelta}{ \mbox{\boldmath $\delta$}}

\newcommand{\bbeta}{\boldsymbol{\beta}}

\newcommand{\blambda}{\boldsymbol{\lambda}}

\newcommand{\bSigma}{\boldsymbol{\Sigma}}

\newcommand{\bZ}{ \mbox{\bf Z}}

\newcommand{\bY}{ \mbox{\bf Y}}

\newcommand{\bI}{ \mbox{\bf I}}

\newcommand{\bW}{ \mbox{\bf W}}

\newcommand{\beqn}{ \begin{eqnarray}}
\newcommand{\eeqn}{ \end{eqnarray}}

\newcommand{\Tr}{^{\rm T}}

\newcommand{\Y}{{\bf Y}}

\newcommand{\ben}{\begin{enumerate}}
\newcommand{\een}{\end{enumerate}}
\newcommand{\beq}{\begin{equation}}
\newcommand{\eeq}{\end{equation}}
\newcommand{\bde}{\begin{description}}
\newcommand{\ede}{\end{description}}

\newcommand{\abs}[1]{\lvert#1\rvert}
\newcommand{\norm}[1]{\lVert#1\rVert}
\newcommand{\norms}[1]{\lVert#1\rVert^2}

\newcommand{\bX}{{\bf X}}

\newcommand{\iid}{\stackrel{\mathrm{iid}}{\sim}}

\newtheorem{theorem}{Theorem}

\newtheorem{lemma}[theorem]{Lemma}

\newcommand{\qn}{\frac{n_2}{n_T}}
\newcommand{\kapatau}{\frac{\tau^2/\sigma_n^2}{1+\tau^2/\sigma_n^2}}
\newcommand{\meanbetats}{\overline{Y}_{Tj} + \kapatau(1-q_n)(\overline{Y}_{2j} - \overline{Y}_{1j})}
\newcommand{\varbetats}{\frac{\sigma^2}{n_2} \frac{q_n + \tau^2/\sigma_n^2}{1+\tau^2/\sigma_n^2}}
\DeclareMathOperator{\EX}{\mathbb{E}}

\newcommand{\Risk}[2]{\mathbb{E}[\|#1 - #2\|^2]}
\newcommand{\Riskrt}[3]{\EX_{#1}[\|#2 - #3\|^2]}
\newcommand{\Reals}[1]{\mathbb{R}^{#1}}
\newcommand{\kll}{D_{\rm{KL}}\left(  \pi_{\tilde{\tau}} \| \pi_0\right)}
\newcommand{\tildetau}{\tilde{\tau}}
\newcommand{\blind}{1}

\begin{document}

\def\spacingset#1{\renewcommand{\baselinestretch}%
{#1}\small\normalsize} \spacingset{1}


\if1\blind
{
  \title{\bf A Bayesian shrinkage estimator for transfer learning}
  \author{Mohamed A. Abba,\thanks{
    This research was partially supported by the National Science Foundation (DMR-2022254) and the National Institutes of Health (R01DE031134 and R56HL155373).  The content is solely the responsibility of the authors and does not necessarily represent the official views of the National Science Foundation nor the National Institutes of Health.}\hspace{.2cm} 
    Jonathan P. Williams,
    and 
    Brian J. Reich\\
    Department of Statistics \\
    North Carolina State University}
  \maketitle
} \fi

\begin{abstract}
Transfer learning (TL) has emerged as a powerful tool to supplement data collected for a target task with data collected for a related source task.  The Bayesian framework is natural for TL because information from the source data can be incorporated in the prior distribution for the target data analysis.  In this paper, we propose and study Bayesian TL methods for the normal-means problem and multiple linear regression. We propose two classes of prior distributions.  The first class assumes the difference in the parameters for the source and target tasks is sparse, i.e., many parameters are shared across tasks. The second assumes that none of the parameters are shared across tasks, but the differences are bounded in $\ell_2$-norm.  For the sparse case, we propose a Bayes shrinkage estimator with theoretical guarantees under mild assumptions. The proposed methodology is tested on synthetic data and outperforms state-of-the-art TL methods. We then use this method to fine-tune the last layer of a neural network model to predict the molecular gap property in a material science application. We report improved performance compared to classical fine tuning and methods using only the target data.
\end{abstract}

\noindent%

\spacingset{1.45} 

\section{Introduction}\label{s:intro}
In the field of statistical inference and modeling, the effective utilization of available data is vital for achieving accurate and robust results, but in many real-world scenarios, obtaining sufficient and high-quality data can be challenging.  {\em TL}, a concept originally popularized in the domain of deep learning \citep[e.g.,][]{yosinski2014transferable,abba2023}, has gained prominence for its potential to address data scarcity and enhance model performance of various machine learning tasks \citep{weiss2016survey}. The fundamental idea of TL is to leverage a task trained from a source domain (associated with a population for which data are readily available) and apply it in training a task in a target domain (associated with a related but distinct population for which limited data are available) to improve predictive performance in the target domain \citep{pan_survey_tl}. 

Our paper establishes theoretical guarantees for TL in the classical {\em normal-means problem}, and naturally extends to high-dimensional linear regression.  We develop a two-stage Bayesian method based on shrinkage priors \citep[e.g.,][]{VANERP201931} in two settings. In the first setting, we assume that the difference between the source and target means is \textit{sparse}, and show under mild conditions that our two-stage method achieves considerable risk reduction compared to using only the target data. In the second setting, we focus on the case where the difference in means is bounded in norm.  

For high-dimensional regression problems, penalized approaches have been shown to provide a safeguard against over-fitting \citep{williams2019,koner2023}, and result in estimators with optimal predictive risk in specific scenarios \citep{mcneish2015using}. Further, penalization fits well within the Bayesian framework, and it is demonstrated in \cite{li2010bayesian, vanderpas_horsehoe}, and \cite{van2017adaptive} that the Bayesian formulation of penalized regression via shrinkage priors leads to at least as good performance as frequentist analogues.  Accordingly, Bayesian shrinkage priors have become a popular tool for inference in high-dimensional problems; see \cite{VANERP201931} for a complete survey.  A notable advantage of the Bayesian approach is the natural uncertainty quantification inherent in posterior probabilistic inference.

Approaches to Bayesian TL for linear regression have been considered previously in the literature, leveraging the convenience to borrow strength across data sets via prior distributions. To our knowledge, however, there has not been any theoretical analysis of Bayesian TL methods for the normal-means or high-dimensional linear regression problems.  The recent Bayesian TL approach developed in \cite{hickey2022} applies to the normal-means and high-dimensional linear regression problems, but focuses on the calibration of prediction sets.  A prior distribution constructed using semi-definite programming is presented in \cite{raina_informative_tl} to build a multivariate Gaussian prior on the target parameters; significant empirical improvement in the target task is reported.  A Bayesian TL framework is proposed in \cite{karbalayghareh2018optimal_bayesian_tl} in the case where the source and target domains are modeled through a joint prior density on the model parameters.  A prior distribution is developed in \cite{Dzyabura2019} to link the source task outcomes to the target, and the method is applied to the online and offline behaviors of customers; this approach could be viewed as a ridge penalty.  Other than that of \cite{hickey2022}, all of these methods require access to the source data, which presents difficulties in many domains with privacy restriction.  Our proposed method only requires access to the parameter estimates from the source task.

Outside of the Bayesian paradigm, a variety of statistical methods for TL in linear regression have been proposed in recent years.  Two-sample hypothesis testing problems are investigated in multiple testing scenarios with a large number of tests in \cite{xia2019gap} and \cite{tony2019two_samples}.  High-dimensional linear models are focused on in \cite{bastani2021predicting} with one source task, and a two-stage estimator is developed with bounded excess risk when the sample size of the source task is larger than the number of covariates.  A nearly optimal method for selection of informative source task samples is constructed in \cite{li2022transfer} in the case where the model parameters are sparse for both the source and target tasks.  Lastly, estimators were developed in \cite{lei2021near} that achieve near-minimax linear risk for linear regression problems under covariate shift.  Aside from our Bayesian construction, another key difference between our method and the aforementioned papers is that we assume no structure on the source task parameters.   

The remainder of the paper proceeds as follows. Sections \ref{s:method} and \ref{s:theory} introduce our proposed statistical methods for the normal-means problem, and establish theoretical properties.  A simulation study is presented in Section \ref{s:sim} to empirically evaluate the performance of our proposed methods compared to other approaches.  In Section \ref{s:app}, we develop an extension to the linear regression case and apply our method to fine-tuning the last layer of a neural network for material informatics prediction.  Our paper concludes in Section \ref{s:discussion} with a few final remarks.

\section{Statistical model}\label{s:method}

Denote the source domain data as $\bY_{11}, \dots, \bY_{1n_1} \in \Reals{p}$ and the target domain data as $\bY_{21}, \dots, \bY_{2n_2} \in \Reals{p}$. We assume the data generating model:
\begin{align*}
    \bY_{11}, \dots, \bY_{1n_1} \mid \bbeta_1, \sigma  & \iid \mbox{Normal}(\bbeta_1, \sigma^2\bI_p), \text{ independent of}\\
    \bY_{21}, \dots, \bY_{2n_2} \mid \bbeta_2, \sigma  & \iid \mbox{Normal}(\bbeta_2, \sigma^2\bI_p),
\end{align*}
where $\bbeta_1$ and $\bbeta_2$ are the $p$-dimensional vectors of means for the source and target data, respectively, and the noise level $\sigma^2$ is assumed to be the same across data sets.  Our primary inferential goal is to estimate the target mean vector $\bbeta_2$ using target data and augmented information from the source data.
 
In the case where the two data sets are similar and $n_2 \ll n_1$ we will argue that the inferential goal is achieved by choosing a prior that penalizes the difference between the means, $\bdelta := \bbeta_2 - \bbeta_1$. If $\bdelta = \mathbf{0}$ the natural prior for $\bbeta_2$ is the posterior of $\bbeta_1$ given $\bY_1=\{\bY_{11},\dots,\bY_{1n_1}\}$. However, if $\bdelta$ is different from zero, this prior would be concentrated around the wrong region of the parameter space. We consider two priors for two different cases for the structure of $\bdelta$. In the first setting (Section \ref{s: means_spase}), we assume that $\bdelta$ is sparse with most entries being exactly zero, implying most elements of $\bbeta_1$ and $\bbeta_2$ are equal. In the second case (Section \ref{s: bounded_means}), the $\ell_2$-norm of $\bdelta$ is assumed to be bounded by some constant, implying none of the elements of $\bbeta_1$ and $\bbeta_2$ are necessarily equal but that their differences are small.

We consider a two-stage approach, where we first fit the source model given the abundant data $\bY_1$ and an improper flat prior for $\bbeta_1$, resulting in the posterior for $\bbeta_1\mid \bY_1, \sigma^2 \sim \mbox{Normal}(\overline{\bY}_1, \sigma^2/n_1\bI_p)$, where $\overline{\bY}_1 := \sum_{i=1}^{n_1} \bY_{1i}/n_1$.  The second stage is to estimate the target mean vector by $\hat{\bbeta}_2 := \hat{\bbeta}_1 + \hat{\bdelta}$, where $\hat{\bbeta}_1 := \EX[\bbeta_1\mid\bY_1]$, i.e., the posterior mean from the source data, and the estimator $\hat{\bdelta}$ is to be determined in the following sections.  We assume that the source data sample size $n_1$ is sufficiently large so that the quality of our target estimate will mainly depend on $\hat{\bdelta}$ and its assumed structure.

\subsection{Sparse case}\label{s: means_spase}

In the case of sparse $\bdelta$, the prior must have considerable mass for nonzero-valued components of $\bdelta$ while still having a spike at the origin for zero-valued components.  This can be achieved using a horseshoe (HS) prior \citep{carvalho2009handling}, although other possibilities are discussed in Section \ref{s:discussion}.  The HS prior belongs to the family of global-local shrinkage priors which have been widely applied for sparse Bayesian learning.  The HS prior is characterized by an infinite mode at the origin and Cauchy-like, heavy tails. This translates to robust, large-signal recovery and adaptation to unknown sparsity \citep{carvalho2009handling}. Furthermore, the theoretical properties of the HS prior have been well-studied \citep{datta2013asymptotic}.  In particular, it is shown in \cite{vanderpas_horsehoe} that when the true parameter vector is sufficiently sparse, the posterior mean constructed with a HS prior achieves the minimax rate and the posterior distribution contracts at the minimax rate. 

Denote the sample mean of the target data by $\overline{\bY}_2 := \sum_{i=1}^{n_2} \bY_{2i}/n_2$.  Given $\sigma$, which can efficiently be estimated using the source data, 
\begin{align*}
\overline{Y}_{1j} \mid  \beta_{1j}, \sigma & \sim \mbox{Normal}(\beta_{1j}, \sigma^2/n_1) \nonumber \\
\overline{Y}_{2j} \mid  \beta_{1j}, \delta_j, \sigma & \sim \mbox{Normal}(\beta_{1j} + \delta_j, \sigma^2/n_2),
\end{align*}
for each vector component $j \in \{1,\dots,p\}$.  Next, let $Z_j := \overline{Y}_{2j} - \overline{Y}_{1j}$ and $\sigma_n^2 := \sigma^2 (\frac{1}{n_1} + \frac{1}{n_2})$ to reparameterize the above model as 
\begin{equation}\label{eq: Z_model}
\begin{split}
\overline{Y}_{1j}  \mid  \beta_{1j}, \sigma & \sim \mbox{Normal}(\beta_{1j}, \sigma^2/n_1) \\
Z_j  \mid  \delta_{j}, \sigma & \sim \mbox{Normal}(\delta_j, \sigma_n^2).
\end{split}
\end{equation}
In this parameterization, $\overline{Y}_{1j}$ and $Z_j$ are sufficient statistics for $\beta_{1j}$ and $\delta_j$, respectively, and so we construct a posterior distribution for $\delta_{j}$ conditional on $\sigma$, based on data generating model \eqref{eq: Z_model}.  In particular, for each $j \in \{1,\dots,p\}$, we specify the HS-like prior
\begin{align}\label{eq: full_hrc_horseshoe}
\delta_j \mid  \lambda_j, \tau, \sigma & \sim \mbox{Normal}(0, \lambda_j^2\tau^2\sigma_n^2) \\
\lambda_j & \sim \mbox{Cauchy}_{+}(0,1). \nonumber 
\end{align}

The prior specification \eqref{eq: full_hrc_horseshoe} on $\delta_j$ depends on two additional scale parameters.  The {\em global shrinkage} parameter $\tau$ can be tailored to be small to encourage sparsity in the posterior concentration for $\bdelta$, while the {\em local} parameters $\lambda_1,\dots,\lambda_p$ allow large signals to escape the global shrinkage.  In the edge case that $\tau=0$ it follows that $\bdelta=\mathbf{0}$ and $\bbeta_{2} = \bbeta_{1}$. In a fully Bayesian specification, an uninformative prior can be placed on the global shrinkage parameter, $\tau$, such as a Half-Cauchy prior, as advocated for in \cite{carvalho2009handling}; a commonly placed uninformative prior on the error variance, $\sigma^2$, is based on an \mbox{Inverse-Gamma} distribution. Alternatively, we can view  $\tau$ as a hyperparameter and estimate it using empirical Bayes methods or marginal maximum-likelihood procedures. Though, it is noted in \cite{vanderpas_horsehoe} that marginal maximum-likelihood estimates run the risk of collapsing to zero, and so the authors propose constraining possible values of $\tau$ to $(1/p, \infty)$. 

Based on the formulation of \eqref{eq: Z_model} and \eqref{eq: full_hrc_horseshoe}:
\begin{align}\label{eq: cond_posterior_delta}
\delta_j \mid Z_j, \lambda_j, \tau, \sigma & \sim \mbox{Normal}\{(1-\kappa_j)Z_j, (1-\kappa_j)\sigma_n^2\},
\end{align}
where $\kappa_j := 1/(1+\lambda_j^2\tau^2)$.  If $\lambda_j$ is large then $\kappa_j$ is close to zero, minimizing the TL; conversely, small $\lambda_j$ implies $\kappa_j$ is close to one and the conditional posterior for $\delta_j$ is nearer a point mass at zero, corresponding to a high degree of TL.  The interpretation of the posterior shrinkage parameter $\kappa_j$ is discussed further in \cite{carvalho2009handling}.  Recalling that $\bdelta := \bbeta_{2} - \bbeta_{1}$, if $\hat{\beta}_{1j}$ denotes a point estimate for $\beta_{1j}$, then a point estimate for $\beta_{2j}$ is $\hat{\beta}_{2j} := \hat{\delta}_j + \hat{\beta}_{1j}$, where $\hat{\delta}_j := \EX[\delta_j \mid  Z_j, \tau, \sigma]$ is a Bayes estimator for $\delta_j$.  

In Section \ref{s:theory}, we derive the asymptotic properties of the Bayes estimator $\hat{\delta}_j$ under assumptions of suitable values of $\tau$.  To do so, we first need to integrate out $\kappa_j$ from the posterior described by \eqref{eq: cond_posterior_delta}. The marginal posterior of $\kappa_j$ conditional on $\tau$ and $\sigma$ has density function
\[
p(\kappa_j \mid  Z_j, \tau, \sigma) \propto (1-\kappa_j)^{-\frac{1}{2}}\frac{1}{1-(1-\tau^2)\kappa_j} e^{-\frac{Z_j^2}{2\sigma_n^2}\kappa_j},
\]
and so with reference to equation \eqref{eq: cond_posterior_delta},
\begin{align} \label{eq: post_mean_delta}
\hat{\delta}_j & = \EX[1-\kappa_j \mid Z_j, \tau, \sigma]\cdot Z_j \nonumber \\
               & = Z_j \cdot \frac{ \int_0^1 (1-\kappa)^{\frac{1}{2}}\frac{1}{1-(1-\tau^2)\kappa} e^{-\frac{Z_j^2}{2\sigma_n^2}\kappa} d\kappa } {\int_0^1 (1-\kappa)^{-\frac{1}{2}}\frac{1}{1-(1-\tau^2)\kappa} e^{-\frac{Z_j^2}{2\sigma_n^2}\kappa} d\kappa} \nonumber \\
               & = Z_j \cdot \frac{ \int_0^1 u^{\frac{1}{2}}\frac{1}{\tau^2+(1-\tau^2)u} e^{-\frac{Z_j^2}{2\sigma_n^2}(1-u)} du } {\int_0^1 u^{-\frac{1}{2}}\frac{1}{\tau^2+(1-\tau^2)u} e^{-\frac{Z_j^2}{2\sigma_n^2}(1-u)} du} \, \mbox{, where  $u = 1-\kappa$} \nonumber \\
               & = Z_j \cdot w_{\tau, \sigma_n^2}(Z_j),
\end{align}
where
\begin{equation}\label{eq: shrinkage_weight}
    w_{\tau, \sigma_n^2}(Z_j) := \frac{ \int_0^1 u^{\frac{1}{2}}g(u) e^{\frac{Z_j^2}{2\sigma_n^2}u} du } {\int_0^1 u^{-\frac{1}{2}}g(u) e^{\frac{Z_j^2}{2\sigma_n^2}u} du}, \quad \mbox{ with } \quad g(u) := \frac{1}{\tau^2+(1-\tau^2)u}.
\end{equation}
Accordingly, the effect of the shrinkage prior on the Bayes estimator is to scale $Z_j = \overline{Y}_{2j} - \overline{Y}_{1j}$ by the weight $w_{\tau, \sigma_n^2}(Z_j) \in (0,1)$, i.e., according to a TL mechanism depending on $\tau$ and $\sigma_{n}$. Furthermore, $w_{\tau, \sigma_n^2}(Z_j)$ is symmetric as a function of $Z_j$ and strictly increasing as a function of the magnitude, $|Z_j|$, for all $j \in \{1,\dots,p\}$.  That being so, the bias in this Bayes estimator $\hat{\delta}_j$ vanishes for large signals because $w_{\tau, \sigma_n^2}(Z_j) \rightarrow 1$ as $|Z_j| \rightarrow \infty$, whereas $\hat{\delta}_j$ shrinks for small signals $|Z_j|$ and can be made arbitrarily small via choosing $\tau$ appropriately close to zero.

For a fully Bayesian treatment of the model, we need to integrate over $\tau$ and $\sigma^2$ in \eqref{eq: post_mean_delta} to get the marginal posterior mean $\EX[\delta_j\mid Z_j]$; to do so, we rely on Markov chain Monte Carlo (MCMC) sampling. It is shown in \cite{van2017adaptive} that the HS prior retains its optimal risk properties under both fully Bayesian estimation and marginal maximum likelihood estimation, with plug-in estimates truncated to the left by $1/p$.  Based on these facts, in a sparse $\bdelta$ scenario, we expect our two-stage estimator to achieve a low error rate on the target task parameters.  Abundant source task data guarantees accurate first-stage estimates.

\subsection{Bounded norm case}\label{s: bounded_means}
In this section, we relax to the assumption that $\bdelta$ need only be bounded in norm, i.e., $\norm{\bdelta} := (\sum_{j=1}^p \delta_j^2)^{1/2} < C$ for some unknown constant $C$.  Since sparsity is no longer assumed, only global shrinkage is warranted, and we consider the global shrinkage prior $\bdelta \mid \tau \sim \mbox{Normal}(\mathbf{0}, \tau^2\bI_p)$.

After the first stage we have the posterior $\beta_{1j} \mid \overline{Y}_{1j}, \sigma \sim \mbox{Normal}(\overline{Y}_{1j}, \sigma^2/n_1)$. This posterior forms the prior in the second-stage estimation:
\begin{align}\label{eq: model_hrc_bounded}
\overline{Y}_{2j} \mid \beta_{2j}, \sigma & \sim \mbox{Normal}(\beta_{2j}, \sigma^2/n_2) \nonumber \\
\beta_{2j} \mid \overline{Y}_{1j}, \sigma, \tau & \sim \mbox{Normal}(\overline{Y}_{1j}, \sigma^2/n_1 + \tau^2).  
\end{align}
The choice of the prior on $\tau^2$ will determine how heavy the tails are as well how sharp the mode is around $\overline{\bY}_1$. To examine the effect of this prior on the posterior of $\bbeta_2$ we first look at the conditional posterior:
\begin{align}
\beta_{2j} \mid  \overline{Y}_{2j}, \overline{Y}_{1j}, \tau & \sim \mbox{Normal}(\mu_j(\tau^2), v_j(\tau^2)) \nonumber \\
\mu_j(\tau^2) & = \meanbetats \nonumber  \\ 
v_j(\tau^2) & = \varbetats, \nonumber 
\end{align}
where $\overline{Y}_{Tj} := (n_1\overline{Y}_{1j} + n_2\overline{Y}_{2j})/n_T$, $n_T := n_1 + n_2$, and $q_n := \qn$. Let $\kappa := \kapatau$. If $\tau^2=0$ then $\kappa = 0$ and the model collapses to pooling the two data sets. On the other hand, if $\tau^2 =\infty$ then $\kappa = 1$ and the model is equivalent to the analysis with only the target data. Hence, $\tau^2$ controls the influence of the source model in the analysis. A prior penalizing large values of $\tau^2$ will in turn penalize deviation from the source estimates.

We use the penalized complexity prior (PCP) \citep{10.1214/16-STS576} on $\tau^2$. In the PCP framework, we first start with a base model, which in our case corresponds to $\tau^2 = 0$, i.e., $\bbeta_2 = \bbeta_1$.  Next, we penalize the deviance of the extended model, corresponding to $\tau^2 >0$, from the base model through an exponential prior on the square root of the Kullback-Leibler (KL) divergence between the two. Let $\pi_0$ and $\pi_\tau$ denote the prior distribution on $\bbeta_2$ in \eqref{eq: model_hrc_bounded} for the base and extended models, respectively. Let $\tilde{\tau}^2 := \tau^2/(\sigma^2/n_1)$, and define the scaled KL divergence between the two models as
\begin{equation}\label{eq: kl_div}
    d(\tildetau^2) := \sqrt{\frac{2}{p} \kll } = \sqrt{\tildetau^2 - \log(1+\tildetau^2)}.
\end{equation}
Clearly $d$ is a strictly increasing function of $\tildetau^2$. Furthermore, if $\tildetau^2=0$ then $d(\tildetau^2) = 0$ and the model collapses to the base model. Assigning an exponential distribution on $d(\tildetau^2) \sim \rm{Exp}(\lambda)$ induces a  prior on $\tildetau^2$ with density
\[
p(d(\tildetau^2))  \times \left| d'(\tildetau^2) \right| = \frac{\lambda\tildetau^2}{1+\tildetau^2} \frac{e^{-\lambda  \sqrt{\tildetau^2 - \log(1+\tildetau^2)}}}{2\sqrt{\tildetau^2 - \log(1+\tildetau^2)}}.
\]

It can be shown that this induced prior on $\tildetau^2$ has a mode at zero with a strictly decreasing density as $\tildetau^2$ increases. Also, the tail of the prior behaves like a modified Weibull distribution \citep{almalki2014modifications} with rate $\lambda$ and shape $0.5$. The hyper-parameter $\lambda$ controls the magnitude of the penalty for deviating from the base model, and can be set using prior belief on the difference between the two tasks.  Properties of the PCP prior in estimating a vector of normal-means are studied in \cite{10.1214/16-STS576}, and it is concluded that PCP priors will assign more mass corresponding to intermediate shrinkage, which is appropriate for cases when the signals are small. Under the current assumption that $\norm{\bdelta}$ is bounded, when $p$ is large all the entries $\delta_j$ must be small, leading to our choice of the PCP prior.

\section{Theoretical properties}\label{s:theory}

In this section, we study the risk of the second-stage estimator in the sparse case. For notational clarity, subscripts of expectation operators specify the parameter value(s) for the data distribution that the expectation is with respect to.  All proofs are given in Appendix A.1.

Let $\bbeta_1^o$ and $\bbeta_2^o$ be the true values of source and target means, respectively, and let $\bdelta^o := \bbeta_2^o - \bbeta_1^o$. First we consider the case where $p = o(n_1)$ and $n_2$ is fixed to reflect a setting with abundant source data but limited target data.  Without loss of generality, assume $n_2 = 1$. In this case the MLE of $\bbeta_1$ is $\hat{\bbeta}_1 = \overline{\bY}_1$ and
\begin{equation}\label{eq: mle_consistency}
\Riskrt{\bbeta_1^o}{\overline{\bY}_1}{\bbeta_1^o} = \frac{p}{n_1}\sigma^2 \longrightarrow 0,
\end{equation}
as $p, n_1 \to \infty$, i.e., the MLE is consistent for estimating the source task parameters.  We will now argue that the risk for estimating $\bbeta_2$ using our two-stage approach attains the minimax risk for estimating a sparse vector of means, where all but $q := \norm{\bdelta^o}_0$ elements of $\bdelta^o$ are exactly zero with $q = o(p)$. Recall that $\hat{\bbeta}_2 = \overline{\bY}_1 + \hat{\bdelta}$, where $\hat{\bdelta} = \EX[\bdelta\mid \bZ, \tau, \sigma]$. 

With the assumption that $q = o(p)$, the minimax risk is attained by $\hat{\bdelta}$ for $\tau = (q/p)^\alpha$ with a suitable choice of $\alpha\geq 1$ \citep{vanderpas_horsehoe}:
\begin{equation}\label{eq: hs_minimax}
\Riskrt{\bdelta^o}{\hat{\bdelta}}{\bdelta^o} \leq \sigma_n^2 \cdot q \log(p/q)\{1 + o(1)\}.
\end{equation}
Accordingly, the expressions in \eqref{eq: mle_consistency} and \eqref{eq: hs_minimax} can be used to express/bound the first two terms in the following expression of the risk of the second-stage estimator for the target task:
\begin{equation}\label{eq: target_risk}
\Riskrt{\bbeta_2^o}{\hat{\bbeta}_2}{\bbeta_2^o} = \frac{p}{n_1}\sigma^2 + \Riskrt{\bdelta^o}{\hat{\bdelta}}{\bdelta^o} + 2\EX_{\bbeta_1^o, \bdelta^o}[ (\overline{\Y}_1  - \bbeta_1^o)^{T} (\hat{\bdelta}  - \bdelta^o)],
\end{equation}
where the rightmost term is strictly negative, as established in Lemma \ref{cl: cross_neg}.
  
\begin{lemma}\label{cl: cross_neg}
Assuming $\overline{\Y}_1 \sim \mbox{Normal}(\bbeta_1^o,\sigma^2I_p)$ independent of $\overline{\Y}_2\sim\mbox{Normal}(\bbeta_2^o,\sigma^2I_p)$ with $\overline{\bdelta} = \overline{\Y}_2-\overline{\Y}_1$ and $\bdelta_0=\bbeta_2^o-\bbeta_1^o$, then
$$\EX_{\bbeta_1^o, \bdelta^o}[ (\overline{\Y}_1  - \bbeta_1^o)^{T} (\hat{\bdelta}  - \bdelta^o)] < 0.$$
\end{lemma}

Our main result, presented as Theorem \ref{thm:main_result} establishes the risk of the estimator.  The following standard assumptions with respect to large $n_1$ and large $p$ are needed:
\begin{itemize}
\item[(A1)] The MLE for the source task is a consistent estimator of $\bbeta_1^o$, implying that $p = o(n_1)$.
\item[(A2)] The proportion of nonzero elements of $\bdelta^o$ vanishes asymptotically, i.e., $q=o(p)$.
\item[(A3)]  $\tau = (q/p)^\alpha$, with $\alpha\geq 1$.
\item[(A4)] The nonzero entries in $\bdelta^0$ have magnitude exceeding the minimum threshold necessary for signal recovery: $\sqrt{2\sigma^2_n\log(p/q)}$.
\end{itemize}

\begin{theorem}\label{thm:main_result}
Under assumptions (A2), (A3), and (A4),
\begin{equation}\label{eq: risk_bound}
\Riskrt{\bbeta_2^o}{\hat{\bbeta}_2}{\bbeta_2^o} \le \frac{p}{n_1}\sigma^2 + \Big[\sigma_n^2 \cdot q \log(p/q) - \frac{\sigma^2}{n_1}\big\{q + (p-q)\tau\sqrt{\log(1/\tau)}\big\}\Big]\{1 + o(1)\}.
\end{equation}
Furthermore, if assumption (A1) is satisfied, then 
\begin{equation}\label{eq: risk_boun_a1}
    \Riskrt{\bbeta_2^o}{\hat{\bbeta}_2}{\bbeta_2^o} \le \sigma_n^2 \cdot q\log(p/q)\{1 + o(1)\}.
\end{equation}
\end{theorem}
It follows from the bound in \eqref{eq: risk_bound} that the overall risk of estimating $\bbeta_2$ is less than the sum of the risks from the two stages (because the third term is negative). The first term on the right is the risk of the MLE of $\bbeta_1$; if $n_1 = O(p)$ then this term is dominated by a constant, and so the risk is dominated by the minimax risk for the estimation of the sparse vector of differences, $\bdelta$.  Next, consider a comparison of the risk of our proposed estimator to the best possible estimator using only the target task. According to \cite{pinsker1980optimal}, the minimax risk for $\bbeta_2$ using only the target data is 
\[
r^* = \frac{p}{n_2}\frac{\sigma^2M^2}{\sigma^2 + M^2} = p\frac{\sigma^2M^2}{\sigma^2 + M^2}, 
\]
for $M$ such that $\norm{\bbeta_2^o} \leq M$.  From \eqref{eq: risk_bound}, it is evident that even in the case where $n_1 = o(p)$ 
\begin{equation}\label{eq: sparse_relative}
    \sup_{\bbeta^o_2}\Risk{\bbeta_2^o}{\hat{\bbeta}_2} = o(r^*).
\end{equation}
Thus, TL of source task information via using an appropriate shrinkage prior leads to an estimator with negligible risk, relative to estimators using only the target data.  Lastly, if the MLE for the source task is consistent, i.e., $p=o(n_1)$, then \eqref{eq: risk_boun_a1} shows that our proposed estimator for $\bbeta_2$ asymptotically attains the minimax risk for a sparse normal-means vector.

\section{Simulation study}\label{s:sim}

Here we apply the methods described in Section \ref{s:method} to simulated data to shed light on the performance of the proposed two-stage TL methods. We illustrate the performance of the methods in the normal-means case for both sparse and bounded settings.

\subsection{Sparse $\bdelta$}

\subsubsection{Data generation}\label{s:s:sim:gen_sparse}
The dimension of the mean vectors will take values $p\in\{100, 500, 10^3, 5\cdot10^3, 10^4\}$ and the number of observations in the target task will take values $n_2\in\{1, 5, 30\}$.  The noise level is $\sigma = 1$ and for the source-task true means $\bbeta^{0}$, the entries will be independently simulated from a uniform distribution on the interval $(-3,3)$. We set the number of nonzero entries in $\bdelta^{0}$ equal to $q := \lceil p^{1-\alpha}\rceil$ for some $\alpha \in (0,1)$ so $q=o(p)$. The entries of $\bdelta^{0}$ follow
\[
\left\{
\begin{array}{ll}
     \delta_{j}^{0} \sim \mbox{Normal}(5, 1)  &   1\leq j \leq q  \\
     \delta_{j}^{0} = 0   &  q < j \leq p. 
\end{array}
\right.
\]
We then set the target means to $\bbeta_2^0 := \bbeta_1^0 + \bdelta^0$ so the source and target model means are the same except for the first $q$ entries. The source task sample size will also vary according to the parameter dimension. We set $n_1 := \lceil p^{1+\gamma}\rceil $, where $\gamma \in (-1, 1)$ so $p=o(n_1)$ when $\gamma$ is positive, $n_1=o(p)$ when $\gamma$ is negative, and $p=O(n_1)$ when $\gamma=0$. Finally, given $\left(\bbeta_1^0, \bdelta, n_1, n_2, \sigma\right)$ we generate the source and target sample means using the following distributions:
\begin{align*}
\overline{Y}_{1j} \sim & \mbox{Normal}(\beta^0_{1j}, \sigma^2/n_1) \\
\overline{Y}_{2j} \sim & \mbox{Normal}(\beta^0_{1j} + \delta^0_j, \sigma^2/n_2).
\end{align*}
We simulated 1,000 data sets for each combination of $p$, $n_2$, and $q$.

\subsubsection{Competing methods and metrics}\label{s:s:sim:methods}

To compare the performance of our proposed estimator, we consider three different estimation methods for the target task parameters, two of which fall within the two-stage framework. In the first stage, we estimate the source data means vector $\hat{\bbeta}_1$, and in the second stage an estimate of the means difference $\hat{\bdelta}$ is computed using $\hat{\bbeta}_1$. For the two-stage methods, we will further examine the effect of the choice of the first-stage estimate. Depending on the value of $\gamma$, we will consider two different choices. If $\gamma >0$, we will only use the MLE $\hat{\bbeta}_1^{ML} = \overline{\bY}_1$ since it is consistent and unbiased. When $\gamma \leq 0$, $\overline{\bY}_1$ will no longer be consistent; in this case we consider both $\hat{\bbeta}_1^{ML}$ and the James-Stein (JS) estimator
 \begin{equation}\label{eq: js_estimator_nm}
\hat{\bbeta_1}^{JS} = \left(1 - \frac{(p-2)\frac{\sigma^2}{n_1}}{\norms{\hat{\bbeta}_1^{ML}}} \right)\hat{\bbeta}_1^{ML}.    
\end{equation}
The JS estimator, although biased, uniformly dominates the MLE and is asymptotically minimax under the assumptions of the Pinsker theorem. 

For the second stage, i.e., the target data, the three different methods, including our proposed estimator for $\bdelta$ are 
\begin{align*}
    \hat{\bdelta}^{ML} & = \bZ \\
    \hat{\bdelta}^{JS} & = \left(1 - \frac{(p-2)\sigma_n^2}{\norms{\bZ}} \right)\overline{\bZ} \\
    \hat{\bdelta}^{HS} & = \{ (1-\EX[\kappa_j \mid Z_j])Z_j \}_{1\leq j \leq p},
\end{align*}
where $Z_j := \overline{Y}_{2j} - \hat{\beta}_{1j}$.  Note that in all cases $\hat{\bdelta}$ depends on the first-stage estimation method. Computational details are given in Appendix A.2.

The main competing method is the Trans-Lasso algorithm \citep{li2022transfer}. The Trans-Lasso algorithm proceeds by first combining the two tasks and estimating a common means vector using the Lasso \citep{tibshirani1996regression} operator. Then uses the Lasso again to estimate the means of the residuals in the target task only. The final Trans-Lasso estimate is the sum of the two-stage estimates.  We compare all the different methods using mean squared error (MSE) for $\bbeta_2$.

\subsection{Bounded $\norm{\bdelta}$}

\subsubsection{Data generation}\label{s:s:sim:gen_bounded}

For the bounded case, the data generation is the same as in Section \ref{s:s:sim:gen_sparse} except for $\bdelta^o$. The true vales of $\bdelta ^o$ are simulated to satisfy $\norm{\bdelta^{0}} \le C = 3\sqrt{p}$. To generate the true values that fall within this $p$-dimensional ball, we first generate $\bW := \{w_i\}_{1\leq i \leq p} \stackrel{iid}{\sim}\mbox{Normal}(0,1)$, and then set $\bdelta^o := CU^{1/p}\bW/\|\bW\|$, where $U\sim \mbox{Uniform}(0,1)$.

\subsubsection{Competing methods}\label{s:s:sim:competing_bounded}

In the bounded case, we consider two competing TL methods along with a target-data-only estimator using the JS method. The main competing TL method is the soft parameter sharing (SPS) method of \cite{yang2023analysis}. The SPS method applies a ridge penalty to learn $\bdelta$, and optimizes the squared error loss for the two tasks jointly.  The properties of SPS are studied in \cite{yang2023analysis} under the bounded difference case, and the authors show that when $\norm{\bdelta} = O(\sigma^2p/n_2)$ the SPS risk is negligible, compared to minimax risk using the target-only data. The second method we consider is to use the JS estimator in the second stage.

\subsection{Results}\label{s:s:sim:results}

These simulation studies suggest that our proposed method outperforms state-of-the-art TL methods in both scenarios. An interesting empirical finding from the simulation study is that an unbiased first-stage estimate is crucial even when the number of parameters is larger than the size of the source task data. In fact, using a biased source task estimate (with a lower risk) results in considerably larger MSE for the target task parameters, particularly in the sparse case. This can be explained by noticing that the original formulation of the priors involves the difference between the true parameter vectors. Hence, when an efficient but biased estimate is used in the first stage, most of the prior mass ends up in the wrong region of the parameter space.

\subsubsection{Sparse $\bdelta$} 

Table \ref{t:sim} summarizes the results for all methods when the source data sample size $n_1$ is larger than the dimension $p$ (i.e., $\gamma>0$). In this case the first-stage estimate is exactly the MLE $\overline{\bY}_1$. We see that as $p$ increases the two-stage HS method outperforms all the other candidates including the Trans-Lasso method.

\begin{table}[H]
\centering
{\footnotesize
\begin{tabular}{lll|cccc|c}
\multicolumn{3}{c}{} & \multicolumn{4}{c}{Methods}  \\
$p$ & $q$ & $n_2$ & MLE & JS & HS &  Trans-Lasso & Max SE\\ \hline
100 & 40 & 1 & 0.99 & 0.91 & 0.76 & 0.82 &  $1.6\cdot10^{-3}$ \\
500 & 145 & 1 & 1.01 & 0.88 & 0.59 & 0.74 & $1.2\cdot10^{-3}$ \\
$10^3$ & 252 & 1 & 1.00 & 0.86 & 0.53 & 0.68 & $1.0\cdot10^{-3}$ \\
$5\cdot10^3$ & 911 & 1 & 1.00 & 0.82 & 0.41 & 0.47 & $0.9\cdot10^{-3}$ \\
$10^4$ & 1585 & 1 & 1.00 & 0.80 & 0.36 & 0.40 &  $0.9\cdot10^{-3}$\\
100 & 40 & 5 & 0.22 & 0.19 & 0.11 & 0.16 &  $1.1\cdot10^{-3}$ \\
500 & 145 & 5 & 0.20 & 0.19 & 0.09 & 0.13 & $1.0\cdot10^{-3}$ \\
$10^3$ & 252 & 5 & 0.20 & 0.18 & 0.08 & 0.12 & $0.8\cdot10^{-3}$ \\
$5\cdot10^3$ & 911 & 5 & 0.19 & 0.18 & 0.07 & 0.10 & $0.8\cdot10^{-3}$ \\
$10^4$ & 1585 & 5 & 0.20 & 0.17 & 0.06 & 0.10 & $0.7\cdot10^{-3}$ \\
100 & 40 & 30 & 0.03 & 0.03 & 0.02 & 0.03 & $5.2\cdot10^{-4}$ \\
500 & 145 & 30 & 0.03 & 0.03 & 0.02 & 0.02 & $4.7\cdot10^{-4}$ \\
$10^3$ & 252 & 30 & 0.03 & 0.03 & 0.02 & 0.02 & $4.8\cdot10^{-4}$ \\
$5\cdot10^3$ & 911 & 30 & 0.03 & 0.03 & 0.01 & 0.02 & $4.1\cdot10^{-4}$ \\
$10^4$ & 1585 & 30 & 0.03 & 0.03 & 0.01 & 0.02 & $2.6\cdot10^{-4}$ \\
\end{tabular}
}
\caption{\footnotesize MSE for $\bbeta_2$ for the sparse case with $p = o(n_1)$, $\alpha=0.2$ and $\gamma = 0.2$.  The scenarios vary by the size of the vector of means ($p$), the number of nonzero means ($q$), and the number of target data samples ($n_2$). The final column gives the largest Monte Carlo standard error in each row to gauge statistical significance between methods.\label{t:sim}} 
\end{table}

In Table \ref{t:sim2} we consider the case where $n_1=p$. In this scenario, the MLE is not a consistent estimator of $\bbeta_1$; hence, we also consider the JS estimator for $\bbeta_1$ for the source task. From the MSE results, it is clear that using an unbiased estimator in the first stage results in a lower MSE for the target data means vector for all the methods. The HS estimator has a lower MSE compared to the other methods, with comparable performance to Trans-Lasso in the very large $p$ cases. 

\begin{table}[H]
\centering
{\footnotesize
\begin{tabular}{lll|ccc|ccc|c|c}
 &&& \multicolumn{3}{c|}{First-stage MLE} & \multicolumn{3}{c|}{First-stage JS} & & \\
$p$ & $q$ & $n_2$ & MLE & JS & HS & MLE & JS & HS & Trans-Lasso & Max SE \\ \hline
100 & 40 & 1 & 0.99 & 0.95 & 0.62 & 0.99 & 7.67 & 13.03& 0.85 & 0.77\\
500 & 145 & 1 & 1.00 & 0.93 & 0.52 & 1.00 & 7.82 & 10.25 & 0.71 & 0.82\\
$10^3$ & 252 & 1 & 1.00 & 0.92 & 0.48 & 1.00 & 7.51 & 8.70 & 0.61 & 0.69 \\
$5\cdot10^3$ & 911 & 1 & 1.00 & 0.90 & 0.39 & 1.00 & 6.44 & 5.41 & 0.58 & 0.73\\
$10^4$ & 1585 & 1 & 1.00 & 0.88 & 0.34 & 1.00 & 5.90 & 4.27& 0.44 & 0.64\\
100 & 40 & 5 & 0.22 & 0.20 & 0.11 & 0.20 & 7.24 & 12.44 & 0.18 & 0.98\\
500 & 145 & 5 & 0.20 & 0.20 & 0.09 & 0.20 & 7.46 & 9.68 & 0.15 & 0.94\\
$10^3$ & 252 & 5 & 0.20 & 0.20 & 0.09 & 0.20 & 7.21 & 8.12 & 0.13 & 0.86 \\
$5\cdot10^3$ & 911 & 5 & 0.19 & 0.20 & 0.07 & 0.20 & 6.21 & 4.71 & 0.12 & 0.57\\
$10^4$ & 1585 & 5 & 0.20 & 0.20 & 0.06 & 0.20 & 5.71 & 3.55 & 0.09 & 0.41\\
100 & 40 & 30 & 0.03 & 0.03 & 0.02 & 0.03 & 7.12 & 12.26 & 0.03 & 1.02\\
500 & 145 & 30 & 0.03 & 0.03 & 0.02 & 0.03 & 7.38 & 9.55 & 0.02 & 0.73\\
$10^3$ & 252 & 30 & 0.03 & 0.03 & 0.01 & 0.03 & 7.13 & 7.97 & 0.02 & 0.59\\
$5\cdot10^3$ & 911 & 30 & 0.03 & 0.03 & 0.01 & 0.03 & 6.15 & 4.55 & 0.02 & 0.53\\
$10^4$ & 1585 & 30 & 0.03 & 0.03 & 0.01 & 0.03 & 4.53 & 1.50 & 0.01 & 0.58\\
\end{tabular}
}
\caption{\footnotesize MSE for $\bbeta_2$ for the sparse case with $p = n_1$ and $\alpha=0.2$.  The scenarios vary by the size of the vector of means ($p$), the number of nonzero means ($q$), and the number of target data samples ($n_2$). The final column gives the largest Monte Carlo standard error in each row to gauge statistical significance between methods.\label{t:sim2}}
\end{table}

Finally, Table \ref{t:sim3} shows the same results when $n_1 = o(p)$. Again, we see that the HS estimator based on the MLE in the first stage achieves the best MSE, despite the latter being inconsistent for $\bbeta_1$ and having an expected risk linear in $p$. The results from the different scenarios indicate that, in terms of MSE, the two-stage approach primarily requires an unbiased estimate of the source data means vector.  The risk in the first-stage estimate tends to have a mitigating effect on the target data risk due to the negative correlation between $\hat{\bdelta}$ and $\hat{\bbeta}_1$ (as established by Lemma 1). 

\begin{table}[H]
\centering
{\footnotesize
\begin{tabular}{lll|ccc|ccc|c|c}
 &&& \multicolumn{3}{c|}{First-stage MLE} & \multicolumn{3}{c|}{First-stage JS} &  \\
$p$ & $q$ & $n_2$ & MLE & JS & HS & MLE & JS & HS & Trans-Lasso & Max SE\\ \hline
100 & 40 & 1 & 0.99 & 0.95 & 0.63 & 1.00 & 7.74 & 13.16 & 0.91 & 0.77 \\
500 & 145 & 1 & 1.00 & 0.94 & 0.52 & 1.00 & 7.81 & 10.23 & 0.82 & 0.87\\
1000 & 252 & 1 & 1.00 & 0.93 & 0.49 & 1.00 & 7.51 & 8.69 & 0.73 & 0.71\\
5000 & 911 & 1 & 1.00 & 0.90 & 0.39 & 1.00 & 6.43 & 5.40 & 0.56 & 0.69\\
10000 & 1585 & 1 & 1.00 & 0.89 & 0.35 & 1.00 & 5.90 & 4.27 & 0.43 & 0.56\\
100 & 40 & 5 & 0.20 & 0.20 & 0.12 & 0.20 & 7.24 & 12.44 & 0.17 & 0.96\\
500 & 145 & 5 & 0.20 & 0.20 & 0.10 & 0.20 & 7.47 & 9.70 & 0.15 & 0.85\\
1000 & 252 & 5 & 0.20 & 0.20 & 0.09 & 0.20 & 7.20 & 8.11 & 0.13 & 0.71\\
5000 & 911 & 5 & 0.20 & 0.20 & 0.07 & 0.20 & 6.21 & 4.72 & 0.10 & 0.68\\
10000 & 1585 & 5 & 0.20 & 0.19 & 0.06 & 0.20 & 5.71 & 3.55 & 0.10 & 0.61\\
100 & 40 & 30 & 0.03 & 0.03 & 0.02 & 0.03 & 7.12 & 12.26 & 0.03 & 1.01\\
500 & 145 & 30 & 0.03 & 0.03 & 0.02 & 0.03 & 7.38 & 9.55 & 0.03 & 0.63\\
1000 & 252 & 30 & 0.03 & 0.03 & 0.01 & 0.03 & 7.13 & 7.97 & 0.02 & 0.62\\
5000 & 911 & 30 & 0.03 & 0.03 & 0.01 & 0.03 & 6.15 & 4.55 & 0.02 & 0.56\\
10000 & 1585 & 30 & 0.03 & 0.03 & 0.01 & 0.03 & 4.53 & 1.50 & 0.01 & 0.47\\
\end{tabular}
}
\caption{\footnotesize MSE for $\bbeta_2$ for the sparse case with $n_1 = o(p)$, $\alpha=0.2$ and $\gamma = -0.2$.  The scenarios vary by the size of the vector of means ($p$), the number of nonzero means ($q$), and the number of target data samples ($n_2$). The final column gives the largest Monte Carlo standard error in each row to gauge statistical significance between methods.\label{t:sim3}}
\end{table}

\subsubsection{Bounded $\norm{\bdelta}$}

Table \ref{t:sim_b} summarizes the results for all methods when the source data sample size $n_1$ is larger than the dimension $p$. In this case the first-stage estimate $\overline{\bY}_1$ is consistent for $\bbeta_1$. We notice that the two-stage PCP method  outperforms all the other candidates, with comparable performance to SPS when $p$ is large. In Table \ref{t:sim_b2} we consider the case where $n_1=p$.  From the MSE results, we notice that the PCP estimator outperforms the other TL methods, but is substantially worse than the JS estimator based on the target data only, especially for the large $p$ cases. So, unlike the sparse case settings, abundant source task data is a requirement for all TL methods in the bounded case.

\begin{table}[H]
\centering
{\footnotesize
\begin{tabular}{ll|ccc|c|c}
 \multicolumn{2}{l}{} & \multicolumn{3}{c}{TL Methods} & Target-only &  \\
$p$  & $n_2$ & SPS & JS & PCP & JS & Max SE\\ \hline
100  & 1 &  1.12 & 1.95 & 0.61 & 0.92 &  0.11 \\
500  & 1 & 1.07 & 1.94 & 0.58 & 0.88 & 0.09 \\
$10^3$ & 1 & 0.96 & 1.93 & 0.52 &  0.89 & 0.13\\
$5\cdot10^3$ & 1 & 1.02 & 1.91 & 0.42 & 0.91 & 0.10\\
$10^4$ & 1 & 0.99 & 1.82 & 0.29 & 0.84 &  0.09\\
100 & 5 &  0.42 & 1.11 & 0.19 & 0.19 &  0.04 \\
500 & 5 & 0.45 & 1.09 & 0.19 & 0.21 & 0.06\\
$10^3$ & 5 & 0.39 & 1.08 & 0.18 &  0.20 & 0.08\\
$5\cdot10^3$ & 5 & 0.31 & 0.97 & 0.18 & 0.21 & 0.08\\
$10^4$ & 5 & 0.25 & 0.96 & 0.17 & 0.19 & 0.06\\
100 & 30 &  0.08 & 1.02 & 0.03 & 0.18 &   0.01\\
500 & 30 & 0.09 & 0.82 & 0.03 & 0.03 & 0.01\\
$10^3$ & 30 & 0.04 & 0.91 & 0.02 &  0.03 & $2.4\cdot10^{-3}$ \\
$5\cdot10^3$ & 30 & 0.03 & 0.81 & 0.02 & 0.03 & $1.8\cdot10^{-3}$ \\
$10^4$ & 30 & 0.03 & 0.74 & 0.02 & 0.03 & $1.6\cdot10^{-3}$ \\
\end{tabular}
}
\caption{\footnotesize MSE for the means $\bbeta_2$ in the bounded case with $p=o(n_1)$ and $\gamma = 0.2$. The scenarios vary by the size of the vector of means ($p$), the number of nonzero means ($q$), and the number of target data samples ($n_2$). The final column gives the largest Monte Carlo standard error in each row to gauge statistical significance between methods.\label{t:sim_b}}
\end{table}

\begin{table}[H]
\centering
{\footnotesize
\begin{tabular}{ll|ccc|c|c}
 \multicolumn{2}{l}{} & \multicolumn{3}{c}{TL Methods} & Target-only &  \\
$p$  & $n_2$ & SPS & JS & PCP & JS & Max SE\\ \hline
100  & 1 &  1.64 & 1.99 & 1.08 & 0.92 &  0.14 \\
500  & 1 & 1.71 & 2.01 & 1.16 & 0.88 & 0.13 \\
$10^3$ & 1 & 1.86 & 2.13 & 1.24 &  0.89 & 0.15\\
$5\cdot10^3$ & 1 & 1.81 & 1.98 & 1.12 & 0.91 & 0.11\\
$10^4$ & 1 & 1.79 & 2.02 & 1.09 & 0.84 &  0.16\\
100 & 5 &  1.09 & 1.21 & 0.99 & 0.19 &  0.09 \\
500 & 5 & 1.15 & 1.14 & 1.01 & 0.21 & 0.09\\
$10^3$ & 5 & 1.13 & 1.18 & 0.98 &  0.20 & 0.07\\
$5\cdot10^3$ & 5 & 1.19 & 1.27 & 0.94 & 0.21 & 0.09\\
$10^4$ & 5 & 1.25 & 1.16 & 0.97 & 0.19 & 0.10\\
100 & 30 &  0.88 & 0.97 & 0.66 & 0.18 &   0.07\\
500 & 30 & 0.79 & 0.92 & 0.48 & 0.03 & 0.07\\
$10^3$ & 30 & 0.82 & 0.84 & 0.52 &  0.03 & 0.07 \\
$5\cdot10^3$ & 30 & 0.93 & 0.72 & 0.51 & 0.03 & 0.07 \\
$10^4$ & 30 & 0.74 & 0.76 & 0.38 & 0.03 & 0.06 \\
\end{tabular}
}
\caption{\footnotesize MSE for the means $\bbeta_2$ in the bounded case with $p=n_1$ and $\gamma = 0$. The scenarios vary by the size of the vector of means ($p$), the number of nonzero means ($q$), and the number of target data samples ($n_2$). The final column gives the largest Monte Carlo standard error in each row to gauge statistical significance between methods.\label{t:sim_b2}}
\end{table}

\section{Material informatics example}\label{s:app}

We apply our proposed TL methods for fine-tuning the last layer of a neural network used for band gap prediction of a molecular crystal target data set. The data are described in \cite{abba2023}.  Band gap, defined as the energy differential between the lowest unoccupied and highest occupied electronic states \citep{kittel2019sons}, determines many aspects of interest in industrial sectors such as electric and photovoltaic conductivity. We use TL methods to borrow information from molecules (source) to molecular crystals (target) to improve prediction. The full target data consists of 50,000 molecular crystals and their band gap values computed using density functional theory \citep[DFT;][]{DFT_citation}. This data set was provided by the Material Science Research Group at Carnegie Mellon University, and required several years to collect. The source data will be the OE62 data set \citep{stuke2020atomic} which consists of $n_1= 62,000$ molecules and their band gap values.  The molecules in the source data were all extracted from organic crystals, and the band gap values were also computed using DFT.  We consider training data consisting of subsets of the target data with $n_2\in\{500, 1{,}000, 2{,}000\}$ to represent a typical materials data set size; we consider a validation set of 10,000 samples, and a testing set of 20,000 samples.

Since we are dealing with molecules, we first need to preprocess the data and compute descriptors. We use the MBTR descriptor method \citep{huo2018unified} to compute $p=1{,}260$ descriptors for every data point in both the source and target data sets. These descriptors are functions of the structure of the material, e.g., the types and configuration of its atoms. A possible approach is to use these features as covariates in a linear regression.  The relationship between these features and band gap, however, is likely not additive or linear.  To capture these complexities, we follow \cite{abba2023} and use a neural network model.  The architecture of the neural network model has two MEGNET layers \citep{chen2019graph} followed by two hidden layers with 1,024 and 512 neurons, respectively. MEGNET is a feature-extraction layer designed specifically for materials problems that exploits physical laws to derive meaningful features from the molecule's structure.

We apply the full deep learning model to the source data, and refit only the output layer parameters using target data.  First, the model is trained on the source data with 40,000 samples for training and 10,000 for validation. The parameters of the best performing model in terms of the MSE on the validation set of the source task are saved and transferred to the target task.  The neural network up to the last hidden layer is used to extract a feature set of size $p=512$ for the target task. 

Let $\hat{\bbeta}_1$ and $\sigma^2\hat{\bSigma}$ be the estimated output layer weights and their estimated variance from the source task, and let $\bbeta_2$ be the weights of the output layer in the target task. We test both the HS and the PCP prior on $\bdelta = \bbeta_2 - \bbeta_1$, and compare their performance to classical ordinary least squares (OLS) estimation using prediction MSE and coverage on the test set. Let $\bY$ denote the true band gap values for the target data, and denote by $\bX$ the extracted features from the last hidden layer of the neural net. For the target data, we fit the following model
\begin{equation}\label{eq: reg_target_model}
\bY \sim \mbox{Normal}(\bX\bbeta_2, \sigma^2\bI_{n_2}).
\end{equation} 
For the PCP and HS methods, we first compute the predicted band gap on the target data as $\tilde{\bY} := \bX\hat{\bbeta}_1 \sim \mbox{Normal}(\bX\bbeta_1, \sigma^2\bX\hat{\bSigma}\bX\Tr)$.  Next, define $\bZ := \bY - \tilde{\bY}$, the difference between the true target data values and the predicted values given the first-stage estimates. We reparameterize \eqref{eq: reg_target_model} in terms of $\bdelta$, and fit the following second-stage regression model
\[
\bZ \sim \mbox{Normal}(\bX\bdelta, \sigma^2(\bI_{n_2} + \bX\hat{\bSigma}\bX\Tr)).
\]
We compare the Bayesian TL methods with simply fitting $\bbeta_2$ on the training data using OLS. Furthermore, a baseline model for comparison that fits the full neural network model using only the target data is also included. 
Due to the difficulty of training an end-to-end neural network model on small amounts of data, when $n_2 = 500$ we first use 5,000 samples for $5$ epochs to get reasonable starting values, then continue training the baseline model on the $500$ target data points. 

For the HS and PCP methods, we run $5\cdot10^5$ iterations of the MCMC sampling algorithm described in Appendix A.2. For the OLS method, we use the usual prediction intervals. For the baseline model, the prediction intervals are obtained using a drop-out layer \citep{srivastava2014dropout} before the last hidden layer of the model.

\begin{table}[H]
\centering
{\footnotesize
\begin{tabular}{c|cccc}
\multicolumn{1}{c}{}& \multicolumn{4}{c}{Method} \\
$n_2$ & PCP & HS & OLS & Baseline \\ \hline
$500$ & $0.91 (83\%)$ & $0.94 (87\%) $ &  $1.08 (89\%)$ & $1.02 (71\%)$ \\
$1000$& $0.83 (92\%)$ & $0.90 (88\%)$  & $0.94 (88\%)$ & $ 0.94 (61\%)$\\
$2000$& $0.81 (91\%) $ &  $0.88 (90\%)$ & $ 0.87 (91\%) $ & $0.82 (63\%)$ 
\end{tabular}
}
\caption{\footnotesize Test set prediction MSE and coverage of $95\%$ intervals for the band gap prediction application with different target data sizes, $n_2$.  The methods are Bayesian TL using PCP and HS priors, TL using OLS, and a baseline deep learning model that uses only target data.}
\label{tab:my_label}
\end{table}

For small target data sets with $n_2=500$, the Bayesian TL methods provide substantial improvements over OLS and the baseline method. As expected, as $n_2$ increases, the advantage of TL dissipates and the prediction accuracy improves for all methods.  For this application, the PCP method generally outperforms the HS method. This could be because all the covariates used in the second stage are actually functions of the same original molecular features, and so all of their effects vary with similar magnitudes between the target and source data sets.

\section{Discussion}\label{s:discussion}

In this paper, we developed a Bayesian estimator for TL in high-dimensional settings. We propose a two-stage procedure based on estimates from the source task, our proposed methodology obviates the need for source task data and only requires the estimates and their uncertainty. This makes our method attractive for cases where the source data is abundant but not available due to privacy rules. In this case, our proposed methodology can leverage prior beliefs on the relationships between similar tasks.

In the Bayesian framework, assumptions on model parameters can naturally be translated into a prior distribution, and so depending on the assumed similarity between source and target tasks we proposed two different estimators, each based on a different prior distribution.  Although we focused on the normal-means model, this work is a step in the direction of further theoretical examination of Bayesian TL methods in linear regression as well as in more complex models such as non-Gaussian and non-continuous response cases. 

When the difference of the means vectors between the two tasks is sparse, a global-local type prior is proposed and the theoretical guarantees for the normal-means model are provided. We observe through theoretical analysis of the risk that the first and second stage of the estimation procedure enjoy a synergy property. Meaning that the total risk of the estimator is strictly less than the sum of the risks from the first- and second-stage estimators. Further, using minimax-type results, we showed that our estimator results in negligible risk, compared to the best estimator based on the target data only. From a TL point of view, this not only means that we run no risk of negative transfer, but also, the proposed estimator vastly outperforms any estimator based solely on the target data. We choose the HS prior for the sparse regression case because it is computationally convenient and has theoretical support.  However, other choices of global-local shrinkage priors \citep[e.g.,][]{bhattacharya2015, bhadra2017, zhang2022bayesian} could result in comparable performance.

\bibliographystyle{agsm}
\bibliography{refs}

@article{bhadra2017,
  title={{The horseshoe+ estimator of ultra-sparse signals}},
  author={Bhadra, Anindya and Datta, Jyotishka and Polson, Nicholas G and Willard, Brandon},
  journal={Bayesian Analysis},
  volume={12},
  number={4},
  pages={1105--1131},
  year={2017},
  publisher={International Society for Bayesian Analysis}
}

@article{bhattacharya2015,
  title={{Dirichlet--Laplace priors for optimal shrinkage}},
  author={Bhattacharya, Anirban and Pati, Debdeep and Pillai, Natesh S and Dunson, David B},
  journal={Journal of the American Statistical Association},
  volume={110},
  number={512},
  pages={1479--1490},
  year={2015},
  publisher={Taylor \& Francis}
}

@article{zhang2022bayesian,
  title={{Bayesian regression using a prior on the model fit: The R2-D2 shrinkage prior}},
  author={Zhang, Yan Dora and Naughton, Brian P and Bondell, Howard D and Reich, Brian J},
  journal={Journal of the American Statistical Association},
  volume={117},
  number={538},
  pages={862--874},
  year={2022},
  publisher={Taylor \& Francis}
}

@article{abba2023,
  title={A penalized complexity prior for deep {B}ayesian transfer learning with application to materials informatics},
  author={Abba, Mohamed A and Williams, Jonathan P and Reich, Brian J},
  journal={Annals of Applied Statistics},
  volume={17},
  number={4},
  pages={3241--3256},
  year={2023},
  publisher={Institute of Mathematical Statistics}
}

@article{hickey2022,
  title={Transfer Learning with Uncertainty Quantification: Random Effect Calibration of Source to Target ({RECaST})},
  author={Hickey, Jimmy and Williams, Jonathan P and Hector, Emily C},
  journal={arXiv preprint arXiv:2211.16557},
  year={2022}
}

@article{tony2019two_samples,
  title={Covariate-assisted ranking and screening for large-scale two-sample inference},
  author={Cai, T. Tony and Sun, Wenguang and Wang, Weinan},
  journal={Journal of the Royal Statistical Society Series B: Statistical Methodology},
  volume={81},
  number={2},
  pages={187--234},
  year={2019},
  publisher={Oxford University Press}
}

@article{karbalayghareh2018optimal_bayesian_tl,
  title={{Optimal Bayesian transfer learning}},
  author={Karbalayghareh, Alireza and Qian, Xiaoning and Dougherty, Edward R},
  journal={IEEE Transactions on Signal Processing},
  volume={66},
  number={14},
  pages={3724--3739},
  year={2018},
  publisher={IEEE}
}

@article{chen2019graph,
  title={Graph networks as a universal machine learning framework for molecules and crystals},
  author={Chen, Chi and Ye, Weike and Zuo, Yunxing and Zheng, Chen and Ong, Shyue Ping},
  journal={Chemistry of Materials},
  volume={31},
  number={9},
  pages={3564--3572},
  year={2019},
  publisher={ACS Publications}
}

@article{xia2019gap,
  title={{GAP}: {A} general framework for information pooling in two-sample sparse inference},
  author={Xia, Yin and Cai, T. Tony and Sun, Wenguang},
  journal={Journal of the American Statistical Association},
  year={2019},
  publisher={Taylor \& Francis}
}

@article{bastani2021predicting,
  title={Predicting with proxies: Transfer learning in high dimension},
  author={Bastani, Hamsa},
  journal={Management Science},
  volume={67},
  number={5},
  pages={2964--2984},
  year={2021},
  publisher={INFORMS}
}

@inproceedings{lei2021near,
  title={Near-optimal linear regression under distribution shift},
  author={Lei, Qi and Hu, Wei and Lee, Jason},
  booktitle={International Conference on Machine Learning},
  pages={6164--6174},
  year={2022},
  organization={PMLR}
}

@article{weiss2016survey,
  title={A survey of transfer learning},
  author={Weiss, Karl and Khoshgoftaar, Taghi M and Wang, DingDing},
  journal={Journal of Big data},
  volume={3},
  number={1},
  pages={1--40},
  year={2016},
  publisher={SpringerOpen}
}

@article{srivastava2014dropout,
  title={{Dropout: a simple way to prevent neural networks from overfitting}},
  author={Srivastava, Nitish and Hinton, Geoffrey and Krizhevsky, Alex and Sutskever, Ilya and Salakhutdinov, Ruslan},
  journal={Journal of Machine Learning Research},
  volume={15},
  number={1},
  pages={1929--1958},
  year={2014},
  publisher={JMLR. org}
}

@article{yosinski2014transferable,
  title={How transferable are features in deep neural networks?},
  author={Yosinski, Jason and Clune, Jeff and Bengio, Yoshua and Lipson, Hod},
  journal={Advances in Neural Information Processing Systems},
  volume={27},
  year={2014}
}

@article{huo2018unified,
      title={Unified Representation of Molecules and Crystals for Machine Learning}, 
      author={Haoyan Huo and Matthias Rupp},
      year={2018},
      note={arXiv:1704.06439}
}

@article{mcneish2015using,
  title={Using {LASSO} for predictor selection and to assuage overfitting: A method long overlooked in behavioral sciences},
  author={McNeish, Daniel M},
  journal={Multivariate Behavioral Research},
  volume={50},
  number={5},
  pages={471--484},
  year={2015},
  publisher={Taylor \& Francis}
}

@article{li2010bayesian,
  title={{The Bayesian elastic net}},
  author={Li, Qing and Lin, Nan},
  journal={Bayesian Analysis},
  volume={5},
  number={1},
  pages={151},
  year={2010},
  publisher={Institute of Mathematical Statistics}
}

@book{kittel2019sons,
  title={Introduction to solid state physics},
  author={Kittel, Charles and McEuen, Paul},
  year={2019},
  publisher={John Wiley \& Sons, Hoboken, NJ}
}

@article{DFT_citation,
  title = {Self-Consistent Equations Including Exchange and Correlation Effects},
  author = {Kohn, W. and Sham, L. J.},
  journal = {Phys. Rev.},
  volume = {140},
  issue = {4A},
  pages = {A1133--A1138},
  numpages = {0},
  year = {1965},
  month = {Nov},
}

@article{makalic2015simple,
  title={A simple sampler for the horseshoe estimator},
  author={Makalic, Enes and Schmidt, Daniel F},
  journal={IEEE Signal Processing Letters},
  volume={23},
  number={1},
  pages={179--182},
  year={2015},
  publisher={IEEE}
}

@article{stuke2020atomic,
  title={Atomic structures and orbital energies of 61,489 crystal-forming organic molecules},
  author={Stuke, Annika and Kunkel, Christian and Golze, Dorothea and Todorovi{\'c}, Milica and Margraf, Johannes T and Reuter, Karsten and Rinke, Patrick and Oberhofer, Harald},
  journal={Scientific data},
  volume={7},
  number={1},
  pages={58},
  year={2020},
  publisher={Nature Publishing Group UK London}
}

@article{VANERP201931,
title = {Shrinkage priors for {B}ayesian penalized regression},
journal = {Journal of Mathematical Psychology},
volume = {89},
pages = {31-50},
year = {2019},
issn = {0022-2496},
doi = {https://doi.org/10.1016/j.jmp.2018.12.004},
url = {https://www.sciencedirect.com/science/article/pii/S0022249618300567},
author = {Sara {van Erp} and Daniel L. Oberski and Joris Mulder},
}

@inproceedings{raina_informative_tl,
author = {Raina, Rajat and Ng, Andrew Y. and Koller, Daphne},
title = {Constructing Informative Priors Using Transfer Learning},
year = {2006},
isbn = {1595933832},
publisher = {Association for Computing Machinery},
address = {New York, NY, USA},
url = {https://doi.org/10.1145/1143844.1143934},
doi = {10.1145/1143844.1143934},
booktitle = {Proceedings of the 23rd International Conference on Machine Learning},
pages = {713–720},
numpages = {8},
location = {Pittsburgh, Pennsylvania, USA},
series = {ICML '06}
}

@article{pan_survey_tl,
  author={Pan, Sinno Jialin and Yang, Qiang},
  journal={IEEE Transactions on Knowledge and Data Engineering}, 
  title={A Survey on Transfer Learning}, 
  year={2010},
  volume={22},
  number={10},
  pages={1345-1359},
  doi={10.1109/TKDE.2009.191}}

@article{Dzyabura2019,
author = {Dzyabura, Daria and Jagabathula, Srikanth and Muller, Eitan},
year = {2019},
month = {01},
pages = {},
title = {Accounting for Discrepancies Between Online and Offline Product Evaluations},
volume = {38},
journal = {Marketing Science},
doi = {10.1287/mksc.2018.1124}
}

@article{almalki2014modifications,
  title={{Modifications of the Weibull distribution: A review}},
  author={Almalki, Saad J and Nadarajah, Saralees},
  journal={Reliability Engineering \& System Safety},
  volume={124},
  pages={32--55},
  year={2014},
  publisher={Elsevier}
}

@article{tibshirani1996regression,
  title={Regression shrinkage and selection via the {LASSO}},
  author={Tibshirani, Robert},
  journal={Journal of the Royal Statistical Society Series B: Statistical Methodology},
  volume={58},
  number={1},
  pages={267--288},
  year={1996},
  publisher={Oxford University Press}
}

@article{yang2023analysis,
  title={Analysis of information transfer from heterogeneous sources via precise high-dimensional asymptotics},
  author={Yang, Fan and Zhang, Hongyang R and Wu, Sen and Su, Weijie J and R{\'e}, Christopher},
  journal={arXiv preprint arXiv:2010.11750},
  year={2023}
}

@inproceedings{carvalho2009handling,
  title={Handling sparsity via the horseshoe},
  author={Carvalho, Carlos M and Polson, Nicholas G and Scott, James G},
  booktitle={Artificial intelligence and statistics},
  pages={73--80},
  year={2009},
  organization={PMLR}
}

@article{datta2013asymptotic,
  title={Asymptotic properties of {B}ayes risk for the Horseshoe prior},
  author={Datta, Jyotishka and Ghosh, Jayanta K},
  journal={Bayesian Analysis},
  volume={8},
  number={1},
  pages={111--132},
  year={2013}
}

@article{van2017adaptive,
author = {St{\'e}phanie van der Pas and Botond Szab{\'o} and Aad van der Vaart},
title = {{Adaptive posterior contraction rates for the horseshoe}},
volume = {11},
journal = {Electronic Journal of Statistics},
number = {2},
publisher = {Institute of Mathematical Statistics and Bernoulli Society},
pages = {3196 -- 3225},
year = {2017},
doi = {10.1214/17-EJS1316},
URL = {https://doi.org/10.1214/17-EJS1316}
}

@Article{li2022transfer,
  author={Li, Sai and Cai, T. Tony and Li, Hongzhe},
  title={{Transfer learning for high‐dimensional linear regression: Prediction, estimation and minimax optimality}},
  journal={Journal of the Royal Statistical Society Series B},
  year=2022,
  volume={84},
  number={1},
  pages={149-173},
  month={February}
}

@article{stein_1981,
author = {Charles M. Stein},
title = {Estimation of the mean of a multivariate normal distribution},
volume = {9},
journal = {Annals of Statistics},
number = {6},
publisher = {Institute of Mathematical Statistics},
pages = {1135 -- 1151},
year = {1981},
doi = {10.1214/aos/1176345632},
URL = {https://doi.org/10.1214/aos/1176345632}
}

@article{vanderpas_horsehoe,
author = {S. L. van der Pas and B. J. K. Kleijn and A. W. van der Vaart},
title = {{The horseshoe estimator: Posterior concentration around nearly black vectors}},
volume = {8},
journal = {Electronic Journal of Statistics},
number = {2},
publisher = {Institute of Mathematical Statistics and Bernoulli Society},
pages = {2585 -- 2618},
year = {2014},
doi = {10.1214/14-EJS962},
URL = {https://doi.org/10.1214/14-EJS962}
}

@article{pinsker1980optimal,
  title={Optimal filtering of square-integrable signals in {G}aussian noise},
  author={Pinsker, Mark Semenovich},
  journal={Problemy Peredachi Informatsii},
  volume={16},
  number={2},
  pages={52--68},
  year={1980},
  publisher={Russian Academy of Sciences, Branch of Informatics, Computer Equipment and~…}
}

@article{10.1214/16-STS576,
author = {Daniel Simpson and H{\aa}vard Rue and Andrea Riebler and Thiago G. Martins and Sigrunn H. S{\o}rbye},
title = {Penalising model component complexity: A principled, practical approach to constructing priors},
volume = {32},
journal = {Statistical Science},
number = {1},
publisher = {Institute of Mathematical Statistics},
pages = {1 -- 28},
year = {2017},
doi = {10.1214/16-STS576},
URL = {https://doi.org/10.1214/16-STS576}
}

@article{williams2019,
  title={Nonpenalized variable selection in high-dimensional linear model settings via generalized fiducial inference},
  author={Williams, Jonathan P and Hannig, Jan},
  journal={Annals of Statistics},
  volume={47},
  number={3},
  pages={1723--1753},
  year={2019}
}

@article{koner2023,
  title={The {EAS} approach to variable selection for multivariate response data in high-dimensional settings},
  author={Koner, Salil and Williams, Jonathan P},
  journal={Electronic Journal of Statistics},
  volume={17},
  number={2},
  pages={1947--1995},
  year={2023}
}

\section*{Appendix A.1: Proofs}\label{s:A1}

\begin{proof}[Proof of Lemma \ref{cl: cross_neg}]

\begin{align}\label{eq: cross_term}
\EX_{\bbeta_1^o, \bdelta^o}[ (\overline{\Y}_1  - \bbeta_1^o)^{T} (\hat{\bdelta}  - \bdelta^o)] & = \sum_{j=1}^p \EX_{\beta_{1j}^o, \delta_j^o} [ (\overline{Y}_{1j} - \beta_{1j}^o)(\hat{\delta}_j - \delta_j^o) ] \nonumber \\
& = \sum_{j=1}^p \EX_{\beta_{1j}^o, \delta_j^o} [ \EX_{\beta_{1j}^o, \delta_j^o}[\overline{Y}_{1j} - \beta_{1j}^o \mid  Z_j ](\hat{\delta}_j - \delta_j^o) ] \nonumber \\
& = \sum_{j=1}^p \EX_{\delta_j^o} [ -\frac{1}{1+n_1}(Z_j - \delta_j^o)(\hat{\delta}_j - \delta_j^o) ] \nonumber \\
& = -\frac{1}{1+n_1}\sum_{j=1}^p \EX_{\delta_j^o} [ (Z_j - \delta_j^o)\{w_{\tau, \sigma_n^2}(Z_j)Z_j - \delta_j^o\} ].
\end{align}
If $\delta_j^o = 0$:
\[
\EX_{\delta_j^o}[(Z_j - \delta_j^o)\{w_{\tau, \sigma_n^2}(Z_j)Z_j - \delta_j^o\}] =  \EX_{\delta_j^o} [ Z_j^2w_{\tau, \sigma_n^2}(Z_j)] > 0.
\]
If $\delta_j^o \neq 0$:
\[
\EX_{\delta_j^o}[(Z_j - \delta_j^o)\{w_{\tau, \sigma_n^2}(Z_j)Z_j - \delta_j^o\}] =  \EX_{\delta_j^o}[Z_jw_{\tau, \sigma_n^2}(Z_j)(Z_j - \delta_j^o)].
\]

Recall that by Stein's Lemma \citep{stein_1981} if $X$ is a normally distributed random variable with mean $\mu$ and variance $\sigma^2$, and if $h$ is a function such that $\EX[h(X)(X-\mu)]$ and $\EX[h'(X)]$  both exist then 
\begin{equation}\label{eq: stein_lemma}
\EX[h(X)(X-\mu)] = \sigma^2\EX[h'(X)]. 
\end{equation}
Let $h(z) = zw_{\tau, \sigma_n^2}(z)$. Clearly, $\EX_{\delta_j^o}[h(Z_j)(Z_j-\delta_j^o)]$ exists since $0<w_{\tau, \sigma_n^2}(z)<1$. In the notation of \eqref{eq: shrinkage_weight},
\begin{equation}\label{eq: h_derivative}
h'(z) = w_{\tau, \sigma_n^2}(z) + \frac{z^2}{\sigma_n^2}\frac{ \int_0^1 u^{\frac{3}{2}}g(u) e^{\frac{z^2}{2\sigma_n^2}u}du \int_0^1 u^{\frac{-1}{2}}g(u) e^{\frac{z^2}{2\sigma_n^2}u}du - 
\left(\int_0^1 u^{\frac{1}{2}}g(u) e^{\frac{z^2}{2\sigma_n^2}u}du\right)^2} {\left(\int_0^1 u^{-\frac{1}{2}}g(u) e^{\frac{z^2}{2\sigma_n^2}u} du\right)^2}.
\end{equation}

It can be verified that $\EX_{\delta_j^o}[h'(Z)]$ is bounded from above from the fact that $u \le 1/u$ for $u \in (0,1)$, and in fact, $h'(z) \leq 1 + z^2/\sigma^2_n$.  Next, the following argument shows that $\EX_{\delta_j^o}[h'(Z)] > 0$.  Hence, $\EX_{\delta_j^o}[h'(Z)]$ exists, and the result of the lemma is established.

By the Cauchy-Schwartz inequality,
\begin{align}
\int_0^1 u^{\frac{1}{2}}g(u) e^{\frac{z^2}{2\sigma_n^2}u}du  = & \int_0^1 u^{\frac{3}{4}}g(u)
^{1/2}e^{\frac{z^2}{4\sigma_n^2}u}u^{\frac{-1}{4}}g(u)
^{1/2}e^{\frac{z^2}{4\sigma_n^2}u}du \nonumber \qquad \qquad \\
\leq & \left(\int_0^1 u^{\frac{3}{2}}g(u)e^{\frac{z^2}{2\sigma_n^2}u}du\int_0^1 u^{\frac{-1}{2}}g(u)e^{\frac{z^2}{2\sigma_n^2}u}du\right)^{1/2}, \nonumber 
\end{align}
and so
\[
\left(\int_0^1 u^{\frac{1}{2}}g(u) e^{\frac{z^2}{2\sigma_n^2}u}du\right)^2 \leq \int_0^1 u^{\frac{3}{2}}g(u)e^{\frac{z^2}{2\sigma_n^2}u}du \int_0^1 u^{\frac{-1}{2}}g(u)e^{\frac{z^2}{2\sigma_n^2}u}du.
\]
Hence, the second term on the right of \eqref{eq: h_derivative} is nonnegative.  By Stein's identity \eqref{eq: stein_lemma},
\[
\EX_{\delta_j^o} [ (Z_j - \delta_j^o)\{w_{\tau, \sigma_n^2}(Z_j)Z_j - \delta_j^o\}] =  \sigma_n^2\EX_{\delta_j^o}[h'(Z_j)] \ge \sigma_n^2 \EX_{\delta_j^o}[w_{\tau, \sigma_n^2}(Z_j)] > 0.
\]


\end{proof}

\begin{lemma}\label{cl: cross_neg_cases} 
Under assumptions (A2), (A3), (A4), and as $p\rightarrow \infty$.
If $\delta_j^o = 0$:
\[
\EX_{\delta_j^o} [ (Z_j - \delta_j^o)\{w_{\tau, \sigma_n^2}(Z_j)Z_j - \delta_j^o\}] \ge \sigma_n^2 \tau \sqrt{\log(1/\tau)}\{1 + o(1)\}.
\]
If $\delta_j^o \neq 0$:
\[
\EX_{\delta_j^o} [ (Z_j - \delta_j^o)\{w_{\tau, \sigma_n^2}(Z_j)Z_j - \delta_j^o\}]  \ge \sigma_n^2\{1 + o(1)\}.
\]
\end{lemma}

\begin{proof}[Proof of Lemma \ref{cl: cross_neg_cases}]
If $\delta_j^o \neq 0$, then as shown at the end of the proof of Lemma \ref{cl: cross_neg},
\[
\EX_{\delta_j^o} [ (Z_j - \delta_j^o)\{w_{\tau, \sigma_n^2}(Z_j)Z_j - \delta_j^o\}] \ge \sigma_n^2 \EX_{\delta_j^o}[w_{\tau, \sigma_n^2}(Z_j)].
\]
By assumption (A4) it must be the case that $\abs{\delta_j^o} \geq \sqrt{2\sigma_n^2\log(p/q)}$, i.e., greater than the minimum threshold necessary for signal recovery \citep{vanderpas_horsehoe}. In this case, we have $\EX_{\delta_j^o}[w_{\tau, \sigma_n^2}(Z_j)] = 1+o(1)$, hence the result.

Next, consider $\delta_j^o = 0$.  As in \cite{vanderpas_horsehoe}, define the integral:
\begin{equation}\label{eq: I_k}
I_k(Z) := \int_0^1 u^k g(u) e^{\frac{Z^2}{2\sigma_n^2}u}du.
\end{equation}
\[
w_{\tau, \sigma_n^2}(Z_j) = \frac{I_{\frac{1}{2}}(Z)}{I_{-\frac{1}{2}}(Z)} \ge \frac{I_{\frac{3}{2}}(Z)}{I_{-\frac{1}{2}}(Z)}  \ge \frac{I_{\frac{3}{2}}(Z)}{I_{-\frac{1}{2}}(Z)} - w^2_{\tau, \sigma_n^2}(Z_j) \ge 0,
\]
where the first inequality holds because $u^{\frac{1}{2}} \ge u^{\frac{3}{2}}$ for $u \in (0,1)$ and the last inequality holds by the application of the Cauchy-Schwarz inequality in the proof of Lemma \ref{cl: cross_neg}.  Then
\[
\EX_{\delta_j^o} [ (Z_j - \delta_j^o)\{w_{\tau, \sigma_n^2}(Z_j)Z_j - \delta_j^o\}] = 
\EX_{\delta_j^o} [ Z_j^2 w_{\tau, \sigma_n^2}(Z_j)] \geq \EX_{\delta_j^o} \left[ Z_j^2 \left( \frac{I_{\frac{3}{2}}(Z_j)}{I_{-\frac{1}{2}}(Z_j)} - w_{\tau, \sigma_n^2}^2(Z_j) \right) \right].
\]
Applying Theorem 3.4 in \cite{vanderpas_horsehoe} and noticing the same lower bound used for the posterior variance gives
\[
\EX_{\delta_j^o} [ (Z_j - \delta_j^o)\{w_{\tau, \sigma_n^2}(Z_j)Z_j - \delta_j^o\}] \ge \sigma_n^2\tau\sqrt{\log(1/\tau)}\{1 + o(1)\}.
\]

\end{proof}

\begin{proof}[Proof of Theorem \ref{thm:main_result}]
Under assumptions (A2), (A3), (A4), equations \eqref{eq: hs_minimax} and \eqref{eq: target_risk} give
\begin{align*}
\Riskrt{\bbeta_2^o}{\hat{\bbeta}_2}{\bbeta_2^o} & \le \frac{p}{n_1}\sigma^2 + \sigma_n^2 \cdot q \log(p/q)\{1 + o(1)\} + 2\EX_{\bbeta_1^o, \bdelta^o}[ (\overline{\Y}_1  - \bbeta_1^o)^{T} (\hat{\bdelta}  - \bdelta^o)] \\
& \le \frac{p}{n_1}\sigma^2 + \Big[\sigma_n^2 \cdot q \log(p/q) - \frac{\sigma^2}{n_1}\big\{q + (p-q)\tau\sqrt{\log(1/\tau)}\big\}\Big]\{1 + o(1)\},
\end{align*}
where the last inequality holds by \eqref{eq: cross_term} and Lemma \ref{cl: cross_neg_cases}.
\end{proof}

\section*{Appendix A.2: Computational details}\label{s:A2}
\subsection*{Gibbs sampling algorithm for the HS shrinkage model}\label{s:s:A2.2}
Sampling from the HS posterior based from the hierarchical model in \eqref{eq: full_hrc_horseshoe} can be complex, since the Half-Cauchy prior on the local and global scales is not conjugate to the normal model. However we can get a conditionally conjugate sampling scheme by introducing data augmentation variables from an inverse gamma distribution as described in \cite{makalic2015simple}. Recall that if $U^2 \mid W \sim \mbox{Inverse-Gamma}(1, 1/W) \mbox{ and } W \sim \mbox{Inverse-Gamma}(1/2, 1) \mbox{ then } U \sim \mbox{Cauchy}_{+}(0,1)$.
For the sampling of the posterior distribution of $(\bdelta, \blambda, \tau^2, \sigma^2 \mid \bZ)$ in the second stage, we augment the model in \eqref{eq: full_hrc_horseshoe}
{\small
\begin{align}
    Z_j\mid \delta_j, \sigma^2_n & \sim \mbox{Normal}(\delta_j, \sigma^2_n) \nonumber 1\leq j\leq p \\
    \delta_j \mid  \lambda_{j}, \tau, \sigma^2 & \sim \mbox{Normal}(0, \lambda_j^2\tau^2\sigma^2_n) \\
    \sigma^2 & \sim \mbox{Inverse-Gamma}(a,b) \nonumber \\
    \lambda_j \mid \zeta_j & \sim \mbox{Inverse-Gamma}(1,1/\zeta_j) \nonumber \\
    \tau \mid \nu & \sim \mbox{Inverse-Gamma}(1, 1/\nu) \nonumber \\
    \zeta_j &\sim \mbox{Inverse-Gamma}(1/2, 1) \nonumber \\
    \nu &\sim \mbox{Inverse-Gamma}(1/2, 1). \nonumber
\end{align}
}
The above hierarchy allows direct sampling from the full conditional using Gibbs sampling by sampling iteratively the following full conditionals.
{\small
\begin{align*}
    \delta_j \mid  \lambda_j, \tau, \sigma_n^2, Z_j & \sim \mbox{Normal}\left(\frac{\tau^2\lambda_j^2}{1+\tau^2\lambda_j^2}Z_j, \sigma^2_n\frac{\tau^2\lambda_j^2}{1+\tau^2\lambda_j^2} \right) \\
    \sigma^2 \mid  \bZ, \bdelta, \blambda, \tau & \sim \mbox{Inverse-Gamma}\left(p+a, \frac{\norm{\bZ - \bdelta}^2  + \sum_{i=1}^p \delta_j^2/(\lambda_j^2\tau^2)}{2} + b\right) \\
    \lambda_j \mid  Z_j, \delta_j, \sigma_n^2, \tau & \sim \mbox{Inverse-Gamma}\left(1, \frac{1}{\zeta_j} +  \frac{\delta_j^2}{2\sigma_n^2\tau^2}\right) \\
    \tau \mid  \bdelta, \blambda, \sigma_n^2 \nu & \sim \mbox{Inverse-Gamma}\left(\frac{p+1}{2}, \frac{1}{\nu} + \frac{1}{2\sigma_n^2}\sum_{i=1}^p \frac{\delta_j^2}{\lambda_j^2}\right) \\
    \zeta_j \mid  \lambda_j & \sim \mbox{Inverse-Gamma}\left(1, 1+\frac{1}{\lambda_j^2}\right) \\
    \nu \mid  \tau & \sim \mbox{Inverse-Gamma}\left(1, 1+\frac{1}{\tau^2}\right).
\end{align*}
}
\subsection*{PCP MCMC sampling algorithm}
For the PCP model, the prior on $\tildetau^2$ is not a conjugate prior, hence we use a Metropolis step to update this parameter. The full MCMC sampling cycles through the following consitional distributions
\begin{align*}
    \bbeta_2 \mid \hat{\bbeta}_1, \overline{\bY}_2, \sigma^2, \tildetau^2 &
    \sim \mbox{Normal}\left( \overline{\bY}_T + \frac{n_2\tildetau^2}{n_T + n_2\tildetau^2}(\overline{\bY}_2 - \overline{\bY}_T), \sigma^2\frac{n_2\tildetau^2}{n_T + n_2\tildetau^2}  \right) \\
    \sigma^2 \mid \hat{\bbeta}_1, \overline{\bY}_2, \tildetau^2 & 
    \sim \mbox{Inverse-Gamma}\left( 2p + a, \frac{n_2\norm{\overline{\bY}_2 - \bbeta_2}^2}{2} + 
    \frac{n_1\norm{\overline{\bY}_1 - \bbeta_2}^2}{2(1+ \tildetau^2)} + b\right) \\
    \pi(\tildetau^2 \mid \bbeta_2, \overline{\bY}_1, \sigma^2) & 
    \propto \frac{\lambda\tildetau^2}{1+\tildetau^2} \frac{e^{-\lambda  \sqrt{\tildetau^2 - \log(1+\tildetau^2)}}}{2\sqrt{\tildetau^2 - \log(1+\tildetau^2)}} \cdot (1+\tildetau^2)^{-p/2} \exp\bigg\{-\frac{n_1\norm{\overline{\bY}_1 - \bbeta_2}^2}{2\sigma^2(1+\tildetau^2)}\bigg\}.
\end{align*}
For $\tildetau^2$, we use an adaptive Metropolis-Hastings step based on the Hessian of the log-posterior.

\end{document}